\documentclass[11pt]{amsart}
\usepackage[margin=1.4in]{geometry}

\usepackage[utf8x]{inputenc}

\usepackage{amsmath}
\usepackage{amssymb}
\usepackage{comment}
\usepackage{amsthm}
\usepackage{graphicx}

\renewcommand{\bar}[1]{\overline{#1}}

\newtheorem{thm}{Theorem}[section]
\newtheorem{theorem}[thm]{Theorem}
\newtheorem*{theorem*}{Theorem}
\newtheorem{lemma}[thm]{Lemma}

\newtheorem{proposition}[thm]{Proposition}
\theoremstyle{definition}
\newtheorem{definition}[thm]{Definition}
\theoremstyle{remark}

\newcommand{\LebNum}{\mathcal{N}}

\newcommand{\dF}{\mathrm{d}_{\mathrm{C}}}
\newcommand{\dH}{\mathrm{d}_{\mathrm{H}}}
\newcommand{\dTV}{\mathrm{d}_{\mathrm{TV}}}
\newcommand{\dbL}{\mathrm{d}_{\mathrm{bL}}}
\newcommand{\BL}{\mathrm{BL}}

\newcommand{\Conv}{\mathrm{C}}
\newcommand{\Class}{\mathcal{C}}
\newcommand{\rotund}{\mathrm{rotund}}
\newcommand{\const}{\mathrm{const}}
\newcommand{\Sph}{\mathcal{S}}
\newcommand{\Rsp}{\mathbb{R}}
\newcommand{\n}{\mathbf{n}}
\newcommand{\B}{\mathrm{B}}
\newcommand{\V}{\mathrm{V}}

\newcommand{\Haus}{\mathcal{H}}
\newcommand{\Prob}{\mathbb{P}}

\newcommand{\abs}[1]{\left\vert #1 \right\vert} 
\newcommand{\nr}[1]{\| #1 \|} 
\renewcommand{\d}{\mathrm{d}}
\newcommand{\dd}{\mathrm{d}}
\newcommand{\h}{\mathrm{h}}
\newcommand{\mean}{\mathrm{mean}}

\newcommand{\eps}{\epsilon}
\newcommand{\sca}[2]            {#1 \cdot #2}
\usepackage{xcolor}

\title[On the reconstruction of convex sets]{On the reconstruction of convex sets \\from random normal measurements}

\author{Hiba Abdallah \and Quentin M\'erigot}

\begin{document}

\thispagestyle{empty}

\begin{abstract}
  We study the problem of reconstructing a convex body using only a
  finite number of measurements of outer normal vectors. More
  precisely, we suppose that the normal vectors are measured at
  independent random locations uniformly distributed along the
  boundary of our convex set. Given a desired Hausdorff error
  $\eta$, we provide an upper bounds on the number of probes that
  one has to perform in order to obtain an $\eta$-approximation of
  this convex set with high probability. Our result rely on the
  stability theory related to Minkowski's theorem.
\end{abstract}

\maketitle 

\section {Introduction}
Surface reconstruction is now a classical and rather well-understood
topic in computational geometry. The input of this problem is a finite
set of points $P$ measured on (or close to) an underlying unknown
surface $S$, and the goal is to reconstruct a Hausdorff approximation
of this surface.
In this article, we deal with a surface reconstruction
question. However, our input is not a set of points, but a set of
(unit outer) normal vectors measured at various unknown locations on
$S$. This question stemmed from a collaboration with CEA-Leti, which
has developed sensors that embed an accelerometer and a magnetometer,
and can return their own orientation, but not their position. Since
these sensors are small and rather inexpensive it is possible to use
many of them to monitor the deformations of a known surface
\cite{Luc}. Can such sensors be used for surface reconstruction\,?
One cannot expect to be able to reconstruct a surface from a finite
number of normal measurements, without assumptions on the surface or
on the distribution of the points. Here, we study the case where the
underlying surface is convex, and where the normals are measured at
random and uniformly distributed locations.

\subsection*{Related work} 
Minkowski's theorem asserts that a convex set is uniquely determined,
up to translation, by the distribution of its normals on the
sphere. In the case of polyhedron, the precise statement is as
follows: given a set of normal vectors $\n_1,\hdots,\n_N$ in the unit
sphere, and a set of positive numbers $a_1,\hdots, a_N$ such that (i)
$\sum_{i=1}^N a_i \n_i = 0$ and (ii) the set of normals spans
$\Rsp^d$, there exists a convex polytope with exactly $N$ faces and
such that the area of the face with normal $\n_i$ is $a_i$
\cite{min2}. Moreover, this polytope is unique up to translation. The
computational aspects related to Minkowski's theorem have been studied
using variational techniques on the primal problem \cite{little} or on
the dual problem \cite{carlier2004theorem,lachand}, but also from the
viewpoint of complexity theory \cite{gritz}.

Minkowski theorem can be interpreted as a reconstruction result, and
such a result comes with a corresponding stability question: if the
distribution of normals of two convex bodies are close to each other
(in a sense to be made precise), are the corresponding bodies also
close in the Hausdorff sense up to translation ? This question has
been studied extensively in the convex geometry literature, using the
theory surrounding the Brunn-Minkowski inequality, starting from two
articles by Diskant \cite{diskant, diskant2, schneider, hug}. Our
probabilistic convergence theorem bears some resemblance with
\cite{gardner2006convergence}, which studies a different inverse
problem in convex geometry, namely reconstructing a convex set from
its brightness function.

\subsection*{Contributions}
In the present paper, we suppose the existence of an underlying convex
body $K$, which is not necessarily a polytope, and from which a
probing device measures unit outer normals. Our input data is a set of
$N$ unit normals $(\n_i)_{1\leq i\leq N}$, which have been measured at
$N$ locations on the boundary $\partial K$ of $K$. These locations
have been chosen randomly and independently, and are uniformly
distributed with respect to the surface area on $\partial K$. Note
that from now on, we assume that \emph{only the measured normals are
  known to us}, and not the locations they were measured at. The
question we consider is the following: given $\eta>0$, what is the
minimum number of such measurements needed so as to be able to
reconstruct with high probability a convex set $L_N$ which is $\eta$
Hausdorff-close to $K$ up to translation?  Denoting $\dH$ the
Hausdorff distance, we prove the following theorem:

\begin{theorem*}[Theorem~\ref{th:sampling}]
  Let $K$ be a bounded convex set with non-empty interior and whose
  boundary has $(d-1)$-area one. Given $p\in (0,1)$ and $\eta > 0$,
  and
\[N \geq \const\left(K,d\right)
\cdot \eta^{\frac{d(1-d)}2 - 2d}\log(1/p)\] random normal
measurements, it is possible to construct a convex body $L_N$ such
that
\[\Prob\left(\min_{x\in \Rsp^d} \dH(x+K,L_N)\leq
  \eta\right)\geq 1-p.\]
\end{theorem*}

In the course of proving this theorem, we introduce a very weak notion
of distance between measures on the unit sphere, which we call the
``convex-dual distance''. This distance is weaker than usual distances
between measures, such as the bounded-Lipschitz\footnote{The
  bounded-Lipschitz distance coincides with the Wasserstein (or
  Earthmover) distance with exponent one when the two measures have
  the same total mass.} or the total variation
distances. Surprisingly, it is nonetheless sufficiently strong to
control the Hausdorff distance between two convex bodies in term of
the convex-dual distances between their distribution of normals, as
shown in Theorem~\ref{th:stab}. This theorem weakens the hypothesis in
the stability results of Diskant \cite{diskant,diskant2} and
Hug--Schneider \cite{hug}.

\subsection*{Notation}
The Euclidean norm and scalar product on $\Rsp^d$ are denoted $\nr{.}$
and $\sca{.}{.}$ respectively. The unit sphere of $\Rsp^d$ is denoted
$\Sph^{d-1}$, and $\B(x,r)$ is the ball centered at a point $x$ with
radius $r$. We call \emph{convex body} a compact convex subsets of the
Euclidean space $\Rsp^d$ with non-empty interior. The boundary of a
convex body $K$ is denoted $\partial K$. Also, we denote $\Haus^d(A)$
the volume of a set $A$, and $\Haus^{d-1}(B)$ the $(d-1)$-Hausdorff
measure of $B$. These notions coincide with the intuitive notions of
volume and surface area in dimension three. A \emph{(non-negative)
  measure} $\mu$ over a metric space $X$ associates to any (Borel)
subset $B$ a non-negative number $\mu(B)$. It should enjoy also the
following additivity property: if $(B_i)$ is a countable family of
disjoint subsets, then $\mu(\cup_i B_i) = \sum_i \mu(B_i)$. We call
$\mu(X)$ the \emph{total mass} of $\mu$. The measure $\mu$ on $X$ is a
\emph{probability measure} if $\mu(X) = 1$. The unit Dirac mass at a
point $x$ of $X$ is the probability measure $\delta_x$ defined by
$\delta_x(B) = 1$ if $x$ belongs to $B$ and $\delta_x(B)=0$ if not.

\section{Minkowski problem and the convex-dual distance}
\label{sec:distance}

The problem originally posed by Minkowski concerns the reconstruction
of a convex polyhedron $P$ from its facet areas $(a_i)_{1\leq i\leq
  N}$ and unit outer normals $(\n_i)_{1\leq i\leq N}$. This data can
be summarized by a measure on the unit sphere, and more precisely by a
linear combination of Dirac masses: $\mu_P = \sum_{1\leq i\leq N} a_i
\delta_{x_i}$.  Minkowski's problem has been generalized to more
general convex bodies by Alexandrov using the notion of \emph{surface
  area measure}.

Recall that given a convex body $K$ and a point $x$ on its boundary, a
unit vector $v$ is a \emph{unit outer normal} if for every point $y$
in $K$, $\sca{(x-y)}{v} \geq 0$. For $\Haus^{d-1}$-almost every point in
$\partial K$, there is a single outer unit normal, which we denote
$\n_K(x)$. The \emph{Gauss map} of $K$ is the map $\n_K: \partial 
K \to \Sph^{d-1}$.
\begin{definition}
  The \emph{surface area measure of $K$} is a measure $\mu_K$ on the
  unit sphere. The measure $\mu_K(B)$ of a (Borel) subset $B$ of the
  sphere $\Sph^{d-1}$ is the $(d-1)$-area of the subset of $\partial
  K$ whose normals lie in $B$. In other words,
\begin{equation}
\label{sur}
\mu_K(B) := \Haus^{d-1}\left( \{x\in \partial K; \n_K(x) \in
  B\}\right) = \Haus^{d-1}(\n_K^{-1}(B)).
\end{equation}
By definition, the total mass $\mu_K(\Sph^{d-1})$ of the surface area
measure is equal to the $(d-1)$-volume of the boundary $\partial
K$. In particular, the surface area of $K$ is a probability measure if
and only if $K$ has \emph{unit surface area}, i.e. $\Haus^{d-1}(\partial
K) = 1$.
\end{definition}

For instance, if $P$ is a convex polyhedron with $k$ $d$-dimensional
facets $F_1,\hdots,F_k$, the unit exterior normal $\n_P(x)$ is well
defined at any point $x$ that lies on the relative interior of one of
these facets. As noted earlier, the surface area measure of the
polyhedron $P$ can then be written as a finite weighted sum of Dirac
masses, $\mu_P = \sum_{i=1}^N \Haus^{d-1}(F_i) \delta_{\n_{F_i}},$
where the unit normal to the $i$th face is denoted $\n_{F_i}$.

Alexandrov's theorem \cite{Ale} generalizes the reconstruction theorem
of Minkowski mentioned in the introduction. It shows that a convex
body is uniquely determined, up to translation, by its surface area
measure. It also gives a characterization of the measures on the
sphere that can occur as surface area measures of convex bodies.

\begin{definition} Given a measure $\mu$ on the unit sphere
  $\Sph^{d-1}$,
  \begin{itemize}
  \item[(i)] the $\emph{mean}$ of $\mu$ is the point of $\Rsp^d$
    defined by $\mean(\mu) := \int_{\Sph^{d-1}} x \dd \mu_K(x)$. The
    measure $\mu$ has \emph{zero mean} if this point lies at the
    origin.
  \item[(ii)] we say that the measure $\mu$ has \emph{non-degenerate
      support} if for every hyperplane $H \subseteq \Rsp^d$, the
    inequality $\mu_K(\Sph^{d-1} \setminus H) > 0$
    holds. Equivalently, $\mu$ has non-degenerate support if and only
    its mass is not entirely contained on a single great circle of the
    sphere.
  \end{itemize}
\end{definition}

\begin{theorem*}[Alexandrov]\label{Ale}
  Given any measure $\mu$ on $\Sph^{d-1}$ with zero mean and
  non-degenerate support, there exists a convex body $K$ whose surface
  area measure $\mu_K$ coincides with $\mu$. Moreover, this convex
  body is unique up to translation.
\end{theorem*}

\subsection{Convex-dual distance}
One of our goals in this article is to refine existing quantitative
estimates of uniqueness in Alexandrov's theorem. In other words, we
want to be able to express the fact that if the surface area measures
$\mu_K$ and $\mu_L$ are close to each other, then the convex bodies
$K$ and $L$ are also close to each other. For this purpose, we
introduce the convex-dual distance, a very weak notion of distance
between measures on the unit sphere. 

The \emph{support function} of a convex body $K\subseteq \Rsp^d$ is a
function $\h_K:\Sph^{d-1} \to\Rsp$ on the unit sphere defined by the
formula $h_K(u) := \max_{x \in K} \sca{x}{u}$.  We will use the
following known fact of convex geometry, whose proof is included for
convenience.

\begin{lemma}
  \label{lem:lipschitz}
  If $K\subseteq \B(0,r)$, the support function $\h_K$ is
  $r$-Lipschitz and $\abs{h_K} \leq r$.
\end{lemma}
\begin{proof}
  Consider $u$ in the unit sphere, and $x$ in $K$ such that $\h_{K}(u)
  = u\cdot x$. For any vector $v$ in the unit sphere, 
\begin{align*}
\h_K(v) = \max_{y\in K} \sca{v}{y}
      &\geq \sca{v}{x}  = \sca{u}{x} + \sca{(v-u)}{x}\\
      &\geq \h_K(u) - \nr{u - v} \nr{x} \\
&\geq \h_K(u) - r \nr{u - v}.
\end{align*}
Swapping $u$ and $v$ gives the Lipschitz bound. Moreover, for $v$ in
$\Sph^{d-1}$, we get by the Cauchy-Schwartz inequality
\begin{equation*}
 \abs{\h_K(v)} = \max_{y\in K} \abs{\sca{v}{y}} \leq  \nr{v} \max_{y\in Y} \nr{y} \leq r  \qedhere
\end{equation*} 
\end{proof}

\begin{definition}
  Given two measures $\mu,\nu$ on $\Sph^{d-1}$, their
  \emph{convex-dual distance} is defined by:
  \begin{equation}
    \label{condis}
    \dF(\mu, \nu) = \max_{K\subseteq \B(0,1)}
    \abs{\int_{\Sph^{d-1}} \h_K \d \mu - \int_{\Sph^{d-1}} \h_K\d\nu},
  \end{equation}
  where the maximum is taken over the set of convex bodies included in
  the unit ball.
\end{definition}

The function $\dF$ defined this way is non-negative and symmetric, and
it is easily seen to satisfy the triangle inequality on the space of
measures on the sphere $\Sph^{d-1}$. However, nothing forbids \emph{a
  priori} that for general measures the distance $\dF(\mu,\nu)$
vanishes while $\mu\neq \nu$. The restriction of $\dF$ to the space of
surface area measures of convex sets satisfies the third axiom of a
distance, i.e. given two convex bodies $K$ and $L$, the distance
$\dF(\mu_K,\mu_L)$ vanishes if and only if $\mu_K=\mu_L$. The proof of
this fact needs additional tools from convex geometry and is postponed
to Lemma~\ref{lemma:eq}.

\subsection{Comparison with other distances}
There are many notions of distances on spaces of measures. In this
paragraph, we compare the convex-dual distance with two of them.  The
\emph{total variation} distance between two measures $\mu$ and $\nu$
on $\Sph^{d-1}$ is defined by
\begin{equation*}
\dTV(\mu,\nu) = \sup_{B\subseteq \Sph^{d-1}} \abs{\mu(B) - \nu(B)},
\end{equation*}
where the supremum is taken on all Borel subsets. The
\emph{bounded-Lipschitz} distance defined by the following supremum,
where $\BL_1$ denotes the set of functions on the unit sphere that are
$1$-Lipschitz and whose absolute value is bounded by one:
\begin{equation*}
\dbL(\mu,\nu) = \sup_{f \in \BL_1} \abs{\int_{\Sph^{d-1}} f \d \mu -
  \int_{\Sph^{d-1}} f\d\nu}.
\end{equation*}
Lemma~\ref{lemma:weak} shows that the convex-dual distance is the
weakest of these three distances. This implies that a stability result
with respect to this distance is stronger than a stability result with
respect to $\dTV$ or $\dbL$. The main advantage for using the
convex-dual distance over the bounded-Lipschitz distance comes from
the fact that the set of support functions of convex sets  included in
$\B(0,1)$ is \emph{much smaller} than the set $\BL_1$. We will show in
Section~\ref{sampling} the implications of this fact on the speed of
convergence of random sampling.

\begin{lemma}
\label{lemma:weak}
Given two measures $\mu,\nu$ on $\Sph^{d-1}$, $\dF(\mu,\nu) \leq
\dbL(\mu,\nu) \leq \const(d) \dTV(\mu,\nu).$
\end{lemma}

\begin{proof}
  By Lemma~\ref{lem:lipschitz}, the support function $\h_K$ of a
  convex set $K$ contained in the ball $\B(0,1)$ is $1$-Lipschitz and
  $\abs{h_K}$ is bounded by one. This implies that $\h_K$ lies in
  $\BL_1$, and therefore
  \[\abs{\int_{\Sph^{d-1}} \h_K \d_\mu - \int_{\Sph^{d-1}} \h_K \d_\nu} \leq \dbL(\mu,\nu).\]
  Taking the maximum over all such support functions gives
  $\dF(\mu,\nu) \leq \dbL(\mu,\nu)$.  The second inequality follows
  from e.g. \cite{dud}, Theorem~6.15.
\end{proof}

\section{Stability in Minkowski Problem }
\label{min}
In this section, we refine existing stability results for Minkowski's
problem so as to obtain a stability result with respect to the
convex-dual distance between surface area measures. We rely and
improve upon existing stability results due to Diskant and
Hug--Schneider, using our definition of convex-dual distance and using
a $\Class^0$ regularity estimate for Minkowski's problem due to Cheng
and Yau.

The following stability theorem is Theorem~3.1 in \cite{hug}, and is
deduced from earlier results of Diskant \cite{diskant,diskant2}, see
also \cite{schneider}. The \emph{inradius} of a convex body $K$ is the
maximum radius of a ball contained in $K$ and the \emph{circumradius}
is the minimum radius of a ball containing $K$.

\begin{theorem*}[Diskant, Hug--Schneider; Theorem~3.1 in \cite{hug}]
  Let $K$ and $L$ be convex bodies with inradius at least $r>0$ and
  circumradius at most $R < +\infty$. Then,
\begin{equation}
  \min_{x\in \mathbb R^d} \dH(K+x, L)\leq \const(r,R,d)
\dbL^{1/d}(\mu_K,\mu_L). 
  \label{eq:hug}
\end{equation}
\end{theorem*}

The main drawback for applying this theorem in the setting of
geometric inference is that one makes an assumption regarding the
inradius and circumradius of the underlying set $K$ \emph{but also} a
similar assumption on the reconstructed set $L$. The second drawback
is that the right-hand side involves the bounded-Lipschitz distance
$\dbL$ instead of the weaker convex-dual distance $\dF$.  Our
improvements to the previous stability results can be summarized as
the the following theorem that obtains the same conclusions with
weaker hypothesis:
\begin{theorem}
\label{th:stab}
Given a convex body $K$ in $\Rsp^d$, and for any measure $\mu$ on
$\Sph^{d-1}$ with zero mean and such that $\dF(\mu_K,\mu) \leq
\eps_0$, there is a convex set $L$ whose surface area measure
coincides with $\mu$ and
\begin{equation}
  \min_{x\in \mathbb R^d} \dH(K+x, L)\leq c
\dF^{1/d}(\mu_K,\mu),
  \label{eq:stab}
\end{equation}
where $c$ and $\eps_0$ are two positive constants depending on $d$ and
$K$ only.
\end{theorem}

As we will see later, the constants in the theorem above depend on the
dimension, on the weak rotundity of the surface area measure $\mu_K$,
defined in the next paragraph, and on the area $\Haus^{d-1}(\partial
K)$.  The exponent in the right-hand side of \eqref{eq:stab} is very
likely not optimal, but the optimal exponent is bounded from below by
$\frac{1}{d-1}$, as noted in \cite{hug}.

%

The remainder of this section is organised as follows. We introduce in
\S\ref{subsec:chengyau} the notion of weak rotundity of the surface
area measure of $K$, and show how the lower and upper bounds on the
inradius and circumradius in Diskant's theorems can be replaced by a
lower bound on the weak rotundity using a lemma of Cheng and
Yau. Then, we recall some known facts from the theory of stability in
Minkowski's theorem in \S\ref{subsec:stab}. Finally, we combine these
results in \S\ref{subsec:proofstab} to get a proof of
Theorem~\ref{th:stab}

\subsection{Weak rotundity}
\label{subsec:chengyau}
In this paragraph, we use a lemma of Cheng and Yau in order to remove
the assumption on the inradius and circumradius of one of the two
convex sets.
%
We call \emph{weak rotundity} of a measure $\mu$ on the unit sphere
the following quantity
\[ \rotund(\mu) := \min_{y \in \Sph^{d-1}}
\left(\int_{\Sph^{d-1}} \max(\sca{y}{v},0) \dd \mu(v)\right) \]
%
Note that the positivity of $\rotund(\mu)$ is equivalent to the
hypothesis that $\mu$ has non-degenerate support. 
\begin{lemma}
\label{lemma:nondeg}
Given a measure on the sphere $\Sph^{d-1}$, $\rotund(\mu) > 0$ if and
only if for any hyperplane $H\subseteq \Rsp^d$ one has
$\mu(\Sph^{d-1}\setminus H) > 0$.
\end{lemma}
\begin{proof}
  If there was a hyperplane $H=\{y\}^\bot$ such that the support of
  $\mu$ is included in $\Sph^{d-1}\cap H$, one would have
\[ \rotund(\mu) \leq \int_{\Sph^{d-1}} \max(\sca{x,y},0) \dd\mu(x) =
\int_{\Sph^{d-1} \cap H} \max(\sca{x,y},0) \dd\mu(x) = 0.\] Therefore,
if $\rotund(\mu) > 0$, the measure $\mu$ must have non-degenerate
support.
\end{proof}

More interestingly, Cheng and Yau \cite{cheng1976regularity}
established a quantitative lower bound on the inradius and an upper
bound on the circumradius of $K$ in term of weak rotundity of the
surface area measure of $K$. Note that in their statement, the
boundary $\partial K$ is assumed to be of class $\Class^4$, but their
proof does not use this fact and can be extended verbatim to the
non-smooth case. A simpler proof of these bounds using John's
ellipsoid is presented in \cite[\S1.1]{guan1998monge}.

\begin{proposition}[Cheng-Yau lemma] Let $K$ be a convex body of
  $\Rsp^d$. Then, the inradius $r$ and circumradius $R$ of $K$ satisfy
  the inequalities:
\begin{align*}
R &\leq \const(d) \left[\mu_K(\Sph^{d-1})\right]^{\frac{d}{d-1}} \rotund(\mu_K)^{-1}, \\
r &\geq \const(d) \left[\mu_K(\Sph^{d-1})\right]^{-d} \rotund(\mu_K)^d.
\end{align*}
\end{proposition}

The advantage of the weak rotundity of $\mu_K$ over the inradius and
circumradius of $K$ is that this quantity is stable with respect to
the convex-dual distance between measures on the sphere.

\begin{lemma}
  Let $\mu,\nu$ be two measures on the unit sphere. Then,
\label{lemma:nondegstab}
\begin{align}
  \abs{\rotund(\mu) - \rotund(\nu)} &\leq \dF(\mu,\nu),
\label{eq:rotstab}\\
  \abs{\mu(\Sph^{d-1}) - \nu(\Sph^{d-1})} &\leq \dF(\mu,\nu).
\label{eq:massstab}
\end{align}
\end{lemma}

\begin{proof}
  We prove Eq.~\eqref{eq:rotstab} first. Given a point $y$ on the unit
  sphere, let $S_y$ denote the line segment joining the origin to
  $y$. Then,
\[ \h_{S_y}(v) = \max_{x \in S_y} \sca{v}{x} = \max(\sca{v}{y},\sca{v}{0}) =\max(\sca{v}{y},0).\]
Define $f_\mu(y) := \int_{\Sph^{d-1}} \max(\sca{y}{v},0) \dd \mu(v)$
and define $f_\nu$ similarly. As a consequence of the definition of
the convex-dual distance, we obtain
\begin{align*}
\abs{f_\mu(y) - f_\nu(y)} &= \abs{\int_{\Sph^{d-1}} \max(\sca{y}{v},0) \dd \mu(v) - \int_{\Sph^{d-1}} \max(\sca{y}{v},0) \dd \mu(v)}\\
&=\abs{\int_{\Sph^{d-1}} \h_{S_y}(v) \dd \mu(v) - \int_{\Sph^{d-1}} \h_{S_y}(v) \dd \mu(v)} \leq \dF(\mu,\nu).
\end{align*}
We have just shown that the uniform distance between the functions
$f_\mu$ and $f_\nu$ is bounded by $\dF(\mu,\nu)$. In particular, the
difference between the minimum of those functions is bounded by the
same quantity, i.e.  $\abs{\rotund(\mu)-\rotund(\nu)} \leq
\dF(\mu,\nu)$.  Inequality \eqref{eq:massstab} is obtained simply by
plugging the support function of the unit ball, $\h_{\B(0,1)} = 1$, in
the definition of the convex-dual distance:
\begin{equation*}
 \abs{\mu(\Sph^{d-1})- \nu(\Sph^{d-1})}  = \abs{\int_{\Sph^{d-1}} \h_{\B(0,1)} \dd(\mu- \nu)} \leq \dF(\mu,\nu)\qedhere
\end{equation*}
\end{proof}

\subsection{Background on stability theory}
\label{subsec:stab}
We need to introduce some tools from convex geometry in order to prove
Theorem~\ref{th:stab}. We make use of the following representations
for the volume $V(K)$ of a convex body $K$ and the first mixed volume
$\V_1(K, L)$ of $K$ with another convex body $L$. The reader can
consider these formulas as definitions. More details on mixed volumes
can be found in e.g. \cite[Chapter~5]{schneider}.
\[
  \V(K)=\frac 1d \int_{\Sph^{d-1}}h_K(u) d\mu_K(u) \qquad
 \V_1(K, L) =\frac 1d \int_{\Sph^{d-1}} h_L(u) d\mu_K(u).
\]
The following inequality is called Minkowski's isoperimetric inequality:
\begin{equation}\label {ineq}\V^d_1(K,L) \geq \V^{d-1}(K) \V(L).\end{equation} 
Equality holds in \eqref{ineq} if and only if the convex sets $K$ and
$L$ are equal up to homothety and translation.  With $L$ equal to the
unit ball, one recovers the usual isoperimetric inequality since then,
$\V_1(K,L) = \frac{1}{d}\Haus^{d-1}(\partial K)$.
Minkowski's inequality implies that the convex-dual distance introduced
in Section~\ref{sec:distance} is indeed a distance between surface
area measures. 
\begin{lemma} $\dF(\mu_K,\mu_L) = 0$ if and only if $\mu_K = \mu_L.$
\label{lemma:eq}
\end{lemma} 
\begin{proof}
  The hypothesis $\dF(\mu_K,\mu_L) = 0$ implies that for any compact
  convex set $M$ contained in the unit ball, one has\begin{equation}
    \int_{\Sph^{d-1}} \h_M \dd \mu_K = \int_{\Sph^{d-1}} \h_M \dd
    \mu_L.
\label{eq:minkapp}
\end{equation}
If one replaces $M$ by $\lambda M$, with $\lambda > 0$, the two sides
of this equality are multiplied by $\lambda$. Thus,
Eq.~\eqref{eq:minkapp} holds for any convex body $M$, regardless of
the assumption that $M$ is contained in $\B(0,1)$.  Taking $M = L$ in
Eq.~\eqref{eq:minkapp} we get
\begin{equation}
\V_1(K,L) = \frac{1}{d} \int_{\Sph^{d-1}} \h_{L}(u) \dd \mu_K(u) = \frac{1}{d}\int_{\Sph^{d-1}} \h_L(u) \dd \mu_L(u) = \V(L).
\end{equation}
Combining this with Minkowski's inequality implies
\begin{equation}
\V(L)^d = \V^d_1(K,L) \geq \V^{d-1}(K) \V(L). \label{eq:minapp}
\end{equation}
Exchanging the role of $K$ and $L$, we see that the volumes of $K$ and
$L$ agree, and the inequality \eqref{eq:minapp} becomes an
equality. Using the equality case in Minkowski's inequality, this
implies that $K$ and $L$ are equal up to homothety and
translation. Using again the equality of volumes of $K$ and $L$, we
see that the factor of the homothety has to be one. Consequently, $K$
and $L$ are translate of each other, and the surface area measures
$\mu_K$ and $\mu_L$ are equal.
\end{proof}

Minkowski's isoperimetric inequality is at the heart of Diskant's
stability results. Instead of using Diskant's theorems directly, we
will use the following consequence \cite[Theorem~7.2.2]{schneider}.

\begin{theorem}[Diskant, Schneider]
\label{th:ds}
Given two positive numbers $r<R$, there exists a positive constant
$c=\const(r,R,d)$ such that for any pair of convex bodies $K,L$ with
inradii at least $r$ and circumradii at most $R$, and
\begin{equation}
 \eps := \max(\abs{V(K) - V_1(K,L)},\abs{V(L) - V_1(L,K)}), \label{eq:as}
\end{equation}
the following inequality holds:
\begin{equation}
 \min_{x\in\Rsp^d} \dH(K,x+L) \leq c \eps^{\frac{1}{d}}.
\label{eq:thds}
\end{equation}
\end{theorem}

Note that there is one difference between the statement of
Theorem~7.2.2 there and the statement given in Theorem~\ref{th:ds}
here however. We replace the strong assumption that the surface area
measures of $K$ and $L$ are close in the total variation sense by a
consequence of this fact, namely Eq.~(7.2.6) there and
Eq.~\eqref{eq:as} here. This weaker assumption is sufficient for the
proof to work, as noted by Hug and Schneider in
\cite[Theorem~3.1]{hug}.

\subsection{Proof of Theorem~\ref{th:stab}}
Assume that $$\dF(\mu_K,\mu) \leq \eps_0 := \min(\frac{1}{2}
\rotund(\mu_K), \frac{1}{2} \mu_K(\Sph^{d-1})).$$ Then, the stability
results of Lemma~\ref{lemma:nondegstab} imply
\begin{align}
0 < \frac{1}{2}\rotund(\mu_K) &\leq \rotund(\mu)\leq 2 \rotund(\mu_K),
\label{eq:ccy:1}\\
0 < \frac{1}{2}\mu_K(\Sph^{d-1}) &\leq \mu(\Sph^{d-1})
\leq 2\mu_K(\Sph^{d-1}).\label{eq:ccy:2}
\end{align}
In particular, by Lemma~\ref{lemma:nondegstab}, the measure $\mu$ has
non-degenerate support. Applying Alexandrov's theorem, there exists a
convex body $L$ such that $\mu=\mu_L$. Cheng and Yau's lemma and
Equations~\eqref{eq:ccy:1} and \eqref{eq:ccy:2} imply that the inradii
$r_K$ and $r_L$ of $K$ and $L$ are bounded from below by a constant
$r$. Similarly, the circumradii $R_K$ and $R_L$ are bounded by a
constant $R$. These constants $r$ and $R$ depend only on
$\rotund(\mu_K)$ and $\Haus^{d-1}(\partial K)$. Now, by definition of
the mixed volumes,
\begin{align}
\abs{\V(K) - \V_1(K,L)}
&= \abs{\int_{\Sph^{d-1}} (\h_K - h_L) \dd \mu_K} \notag \\
&\leq \abs{\int_{\Sph^{d-1}} \h_K \dd (\mu_K-\mu_L)} + 
\abs{\int h_L  \dd (\mu_K-\mu_L)}.
\label{eq:st:1}
\end{align} 
Since the stability theorem we are proving is up to translations, we
can translate $K$ and $L$ if necessary.  The circumradii $R_K$ and
$R_L$ are bounded by $R$, and we therefore assume that $K$ and $L$ are
included in the ball $\B(0,R)$. This means that the bodies $K' =
\frac{1}{R} K$ and $L' = \frac{1}{R} L$ are included in the unit
ball. Note also that $\h_{L'} = R \h_L$. Putting the definition of the
convex-dual distance into Eq.~\eqref{eq:st:1}, this gives
\[
\abs{\V(K) - \V_1(K,L)}
\leq R \abs{\int_{\Sph^{d-1}} \h_{K'} \dd (\mu_K-\mu_L)} + 
R \abs{\int h_{L'}  \dd (\mu_K-\mu_L)} 
\leq 2 R \dF(\mu_K,\mu_L).
\]
The same inequality where $L$ and $K$ have been exchanged also holds,
and this allows us to apply Theorem~\ref{th:ds} with $\eps = 2R
\dF(\mu_K,\mu_L)$. Note that the constant that occurs in Eq.
\eqref{eq:thds} depends on quantities that depend on $\rotund(K)$,
$\mu_K(\Sph^{d-1}) = \Haus^d(\partial K)$ and $d$.
\label{subsec:proofstab}

\section{Random sampling}
\label {sampling}
Let $K$ be a convex body and $\mu_K$ its surface area measure. Note
that by measuring normals only, one cannot determine the area of
$\partial K$. Therefore, we assume that $K$ has unit surface area,
i.e. $\mu_K$ is a probability measure. We call \emph{random normal
  measurements} a family of unit vectors $(\n_i)_{1\leq i\leq N}$ that
are obtained by measuring the unit outer normal at $N$ random
independent locations on $\partial K$, whose distribution is given by
the surface area on $\partial K$.  Equivalently, the vectors
$(\n_i)_{1\leq i\leq N}$ are obtained by i.i.d. sampling from the
probability measure $\mu_K$. The \emph{empirical measure} associated
to $\mu_K$ is therefore defined by the formula $\mu_{K,N} := \frac 1N
\sum_{i=1}^N \delta_{\n_i}$. The main result of the article is the
following theorem.

\begin{theorem}\label{th:sampling}
  Let $K$ be a convex body with unit surface area. Given a desired
  probability $p\in (0,1)$, a desired error $\eta > 0$, and given
\[N \geq \const\left(K,d\right)
\cdot \eta^{\frac{d(1-d)}2 - 2d}\log(1/p)\] random normal
measurements it is possible to construct a convex body $L_N$ such that
\[\Prob\left(\min_{x\in \Rsp^d} \dH(x+K,L_N)\leq
  \eta\right)\geq 1-p.\]
\end{theorem}

The exponents that we obtain are $N = \Omega(\eta^{-5})$ in dimension
two and $N = \Omega(\eta^{-9})$ in dimension three, and are most
likely not optimal.

\subsection{Zero-mean assumption}
Note that even if the mean of the the empirical measure $\mu_{K,N}$
will be close to zero with high probability, it will usually not be
exactly zero. However, this is a necessary condition for the existence
of a convex polytope $L$ such that $\mu_L= \mu_{K,N}$. The following
proposition shows that this equality can be enforced without
perturbing $\mu_N$ too much in the sense of the convex-dual distance
$\dF$.

\begin{proposition}\label{prop:perturb}
  Given any convex body $K$ with unit surface area, and any
  probability measure $\nu$ on $\Sph^{d-1}$, there exists a
  probability measure $\bar{\nu}$ on $\Sph^{d-1}$ with zero mean such
  that $\dF(\bar{\nu},\mu_K) \leq 3 \dF(\nu,\mu_K)$.
\end{proposition}

\begin{lemma}
  Given any probability measure $\nu$ on $\Sph^{d-1}$ with mean $m$,
  there exists a probability measure $\bar{\nu}$ on $\Sph^{d-1}$ with
  zero mean such that $\dF(\nu,\bar{\nu}) \leq 2\nr{m}$.
\label{lemma:bar}
\end{lemma}
\begin{proof} 
  We only deal with the case of a probability measure on a finite set
  $\nu = \frac{1}{N} \sum_{1\leq i\leq N} \delta_{x_i}$. The general
  case can be obtained using the density of these measures in the
  space of probability measures. Let $m$ denote the mean of $\nu$,
  i.e. $m = \frac{1}{N} \sum_{1\leq i\leq N} x_i$. By convexity of
  $\B(0,1)$, the point $m$ always lies inside the ball
  $\B(0,1)$. Moreover, by strict strict convexity of the ball,
  $\nr{m}=1$ occurs only when $\nu = \delta_m$. In this case, one can
  check that if $\bar{\nu}$ is the uniform probability measure on
  $\Sph^{d-1}$ then $\dF(\nu,\bar{\nu}) \leq 2$. We will assume from
  now on that $\nr{m} < 1$ and set $\bar{\nu} = \frac 1N \sum_{1\leq
    i\leq N}\lambda a_i \cdot \delta_{m_i},$ where
\[m_i=
\frac{x_i-m}{\nr{x_i-m}}, a_i= \|x_i-m\|, \lambda =
\left(\frac{1}{N}\sum_{1\leq i\leq N} a_i\right)^{-1} .\] By
construction, the measure $\bar{\nu}$ is a probability measure; and it
has zero mean:
\[ \frac{1}{N}\sum_{1\leq i\leq N} \lambda a_i m_i = \frac{1}{N}
\sum_{1 \leq i \leq N} \lambda a_i \frac{x_i - m}{a_i} =
\left(\frac{1}{N} \lambda \sum_{1 \leq i \leq N} x_i \right) - \lambda
m = 0.\] Second, we want to bound the convex-dual distance between
$\nu$ and $\bar{\nu}$. For that purpose, we consider a convex set $M$
included in the unit ball $\B(0,1)$ and $\h_M$ its support
function. We let $\hbar_M$ be the extension of the support function to
$\Rsp^d$ by the same formula $\hbar_M(x) = \sup_{p\in M} \sca{x}{p}$. This
function is positively homogeneous, i.e. for $\lambda>0$,
$\hbar_M(\lambda v) = \lambda \hbar_M(v)$. We have:
\begin{equation}
\begin{aligned}
\abs{\int_{\Sph^{d-1}} \h_M(v) \dd (\nu-\bar{\nu})(v)} &= \frac{1}{N} \abs{\sum_{1\leq i\leq N} \h_M(x_i)  - \sum_{1\leq i \leq N} \lambda a_i \h_M\left(\frac{x_i - m}{a_i}\right)} \\
&\leq \frac{1}{N} \sum_{1\leq i\leq N} \abs{\hbar_M(x_i) - \lambda a_i  \hbar_M\left(\frac{x_i - m}{a_i}\right)}\\
&\leq \frac{1}{N} \sum_{1\leq i\leq N} \abs{\hbar_M(x_i) - \lambda\hbar_M\left(x_i - m\right)}
\label{eq:cdb1}
\end{aligned}
\end{equation}
From the first to the second line we used the triangle inequality, and
from the second to the third line we used the homogeneity of
$\hbar$. Finally, since $M$ is contained in the unit ball, the
function $\hbar_M$ is $1$-Lipschitz (this follows from the same proof
as in Lemma~\ref{lem:lipschitz}). Combining with $\hbar_M(0) = 0$, we get:
\begin{align*}
\abs{\hbar_M(x_i) - \lambda\hbar_M\left(x_i - m\right)} 
  &\leq \abs{\hbar_M(x_i) - \hbar_M\left(x_i - m\right)} 
    + \abs{(1-\lambda) \hbar_M(x_i-m)} \\
  & \leq \nr{m} + \abs{1-\lambda} \nr{x_i-m}
\end{align*}
Summing these inequalities, and using the definition of $\lambda$
gives us
\begin{equation}
\begin{aligned}
\frac{1}{N} \sum_{1\leq i\leq N} \abs{\hbar_M(x_i) - \lambda\hbar_M\left(x_i - m\right)}
&\leq \nr{m} + \abs{1-\lambda}\left(\frac{1}{N}\sum_{i=1}^N \nr{x_i-m}\right) \\
&= \nr{m} + \abs{1-\frac{1}{N}\sum_{i=1}^N \nr{x_i-m}}\leq 2\nr{m}
\label{eq:cdb2}
\end{aligned}
\end{equation}
We conclude using the definition of the convex-dual
distance and Eqs~\eqref{eq:cdb1}--\eqref{eq:cdb2}.
\end{proof}

\begin{proof}[Proof of Proposition~\ref{prop:perturb}]
  We need to show that the mean $m$ of the measure $\nu$ is not too
  far from zero. Given any point $x$ on $\Sph^{d-1}$ and $K_x = \{x\}$
  the convex set consisting of only $x$, one has $\h_{K_x}(v) :=
  \max_{z\in K_x} \sca{z}{v} = \sca{x}{v}$. Therefore, using the
  definition of the convex-dual distance and the fact that $\mu_K$ has
  zero mean we obtain
\begin{align*}
\dF(\mu_K, \nu) &\geq \abs{\int_{\Sph^{d-1}} \sca{x}{v} \d\mu_K(v) - \int_{\Sph^{d-1}} \sca{x}{v}\d\nu(v))} \\
&=\abs{\int_{\Sph^{d-1}} \sca{x}{v}\d\nu(v)} =
\left|\sca{m}{v}\right| .
\end{align*}
Taking $x=m/\nr{m}$ in this inequality proves that $\nr{m}$ is bounded
by $\dF(\mu_K, \nu)$. We can then apply Lemma~\ref{lemma:bar} to
construct $\bar{\nu}$. Using the triangle inequality
for $\dF$ and $\dF(\bar{\nu},\nu) \leq 2\nr{m}$, we get.
\begin{equation*}
\dF(\bar{\nu},\mu_K) \leq \dF(\bar{\nu},\nu) + \dF(\nu,\mu_K) \leq 2 \nr{m} + \dF(\nu,\mu_K) \leq 3\dF(\nu,\mu_K).\qedhere
\end{equation*}
\end{proof}

\subsection{Convergence of the empirical measure}
We consider a probability measure $\mu$ on the unit sphere, and we
denote by $\mu_N$ the empirical measure constructed from $\mu$,
i.e. $\mu_N = \frac{1}{N} \sum_{1\leq i \leq N} \delta_{X_i}$ where
$X_i$ are i.i.d random vectors with distribution $\mu$.  The following
probabilistic statement determines the speed of convergence of $\mu_N$
to $\mu$ for the convex-dual distance.

\begin{proposition}
\label{prop:empirical}
Let $\mu$ be a probability measure on $\Sph^{d-1}$, and $\mu_N$ the
corresponding empirical measure. Then, $\mu_N$ converges to $\mu$ for
the convex dual distance with high probability. More precisely, for
any positive $\eps \leq \const(d)$ and any $N$, the following inequality
holds:
 \begin{equation}
  \label{prob}
 \mathbb P\left[ \dF(\mu_N, \mu_K)\leq \epsilon\right] \geq 1-2 
 \exp \left(\const(d)\epsilon^{\frac{1-d}2}-N\epsilon^2/2
 \right).
 \end{equation}
\end{proposition}

The proof of this proposition relies on the combination of a Theorem
of Bronshtein \cite[Theorem~5]{bro} with Chernoff's bound.  Recall
that the \emph{$\eps$-covering number} $\LebNum(X, \eps)$ of a metric
space $X$ is the minimal number of closed balls of radius $\eps$
needed to cover $X$.  Let $\Conv_1$ be the set of convex bodies
contained in the unit ball $\Rsp^d$, endowed with the Hausdorff
distance.

\begin{theorem*}[Bronshtein] Assuming $\eps \leq \eps_d :=
  10^{-12}/(d-1)$, the following bound holds:
 \[\log_2 (\LebNum(\Conv_1, \epsilon)) \leq \const(d)
  \epsilon^{\frac{1-d}{2}}.\]
\end{theorem*}

\begin{proof}[Proof of Proposition \ref{prop:empirical}]
  By the theorem of Bronshtein, given any positive number $\eps$
  smaller than $\eps_d$, there exists $n$ and $n$ convex body $K_1,\hdots,K_n$
  included in the unit ball such that for any convex body $M \subseteq
  \B(0,1)$ one has $\min_{1\leq i\leq n} \dH(K_i,M) \leq
  \eps$. Moreover, the number $n$ can be chosen smaller than
\begin{equation}
 n = \LebNum(C_1,\eps) \leq \exp\left(\const(d) \eps^{\frac{1-d}{2}}\right). 
\label{eq:br}
\end{equation}
We consider $N$ i.i.d. random points $X_1,\hdots,X_N$ on the unit
sphere whose distribution is given by the measure $\mu$. For a fixed
$i$, the support function $\h_{K_i}$ is bounded by one by
Lemma~\ref{lem:lipschitz}, and one can apply Hoeffding's inequality to
the random variables $(\h_{K_i}(X_k))_{1\leq k\leq N}$. By definition
of the empirical measure $\mu_N$, this gives
\begin{equation} \label{hoff}
\mathbb P\left(\left|\int_{\Sph^{d-1}}
        \h_{K_i}(x) \dd  \mu_N(x) - \int_{\Sph^{d-1}} \h_{K_i} \dd\mu
      \right|\geq \epsilon\right)\leq 2 \exp (-2 N\epsilon^2).
\end{equation}
Taking the union bound, we get
\begin{equation} \label{eq:hoffub}
\mathbb P\left(\max_{1\leq i\leq n} \left|\int_{\Sph^{d-1}}
        \h_{K_i}(x) \dd  (\mu_N-\mu)(x)\right|\geq \epsilon\right)\leq 2 n \exp (-2 N\epsilon^2)
\end{equation}
Now, given any convex body $M$ included in the unit ball, there exists
$i$ in $\{ 1,\hdots, n\}$ such that the distance $\nr{\h_M - \h_{K_i}}
= \dH(M,K_i)$ is at most $\eps$. Thus, for any probability measure
$\nu$ on the sphere,
\[\left|\int_{\Sph^{d-1}}
        \h_{K_i}(x) \dd \nu(x) - \int_{\Sph^{d-1}} \h_{M} \dd\nu
      \right| \leq  \abs{\int_{\Sph^{d-1}} \nr{\h_{K_i}- \h_M}_\infty \dd\nu}
\leq \eps,\]
and as a consequence,
\begin{equation}
\left|\int_{\Sph^{d-1}}
        \h_{M}(x) \dd (\mu_N-\mu)(x) \right|
\leq \max_{1\leq i\leq n} \left|\int_{\Sph^{d-1}}
        \h_{K_i}(x) \dd (\mu_N-\mu)(x) 
      \right| + 2 \eps.
\label{eq:cov}
\end{equation}
The combination of inequalities \eqref{eq:hoffub} and \eqref{eq:cov}
imply that
\begin{equation}
\Prob(\d_C(\mu_{N},\mu)\geq 3\eps) \leq 2 \exp(\log(n) - 2N\eps^2)
\end{equation}
Using the upper bound on $n$ from Eq.~\eqref{eq:br} concludes the proof.
\end{proof}

\subsection{Proof of Theorem \ref{th:sampling}} We assume first that
$\dF(\mu_K,\mu_{K,N})$ is small enough, and more precisely that
$\dF(\mu_K,\mu_{K,N}) \leq \frac{1}{3}\eps_0$, where $\eps_0$ is the
constant given by Theorem~\ref{th:stab}. By
Proposition~\ref{prop:perturb} we can construct a probability measure
$\bar{\mu}_{K,N}$ with zero mean such that
\[\dF(\mu_K,\bar{\mu}_{K,N}) \leq 3 \dF(\mu_K,\mu_{K,N}) \leq \eps_0.\]
This allows us to apply Theorem~\ref{th:stab} to the measure
$\bar{\mu}_{K,N}$. There exists a convex body $L_N$ whose surface area
measure $\mu_{L_N}$ coincides with $\bar{\mu}_{K,N}$ and moreover,
\[ \min_{x\in \Rsp^d} \dH(x+K,L_N) \leq c \dF(\mu_K,\mu_{L_N})^{\frac{1}{d}} = 3^{\frac{1}{d}} c \dF(\mu_{K},\mu_{K,N})^{\frac{1}{d}}. \]
Therefore, using Proposition~\ref{prop:empirical}, and assuming $\eps \leq \frac{1}{3} \eps_0$, we have 
\begin{align*}
  \Prob\left[ \min_{x\in \Rsp^d} \dH(x+K,L_N) \leq 3^{\frac{1}{d}} c
    \eps^{\frac{1}{d}}\right]
    &\geq \Prob\left[ \dF(\mu_K, \mu_{K,N})\leq \epsilon\right]\\
    &\geq 1-2 \exp \left(\const(d)\epsilon^{\frac{1-d}2}-N\epsilon^2/2
    \right).
\end{align*}
Finally, we set $\eta = 3^{\frac{1}{d}} c \eps^{\frac{1}{d}}$, we get
\[\Prob\left[ \min_{x\in \Rsp^d} \dH(x+K,L_N) \leq \eta\right]
\geq 1-2 \exp \left[C \cdot \left(\eta^{\frac{d(1-d)}2}-N\eta^{2d}
  \right)\right].\] for some constant $C$ that only depends on $d$ and
$K$, thus concluding the proof.

\section{Special case: polyhedra}

When the underlying convex body is a convex polyhedron, one can get
much better probabilistic bounds on the speed of convergence. This
model is quite simplistic, however, because of the assumptions that
each of the measured normals must coincide with the one of normals of
the underlying polyhedron. In particular, one cannot hope to extend
this result to handle noise. The proof of this proposition relies on
Theorem~2.1 of \cite{hug} and on a lemma of Devroye.

\begin{figure}
\centering
\includegraphics[height=.4\linewidth]{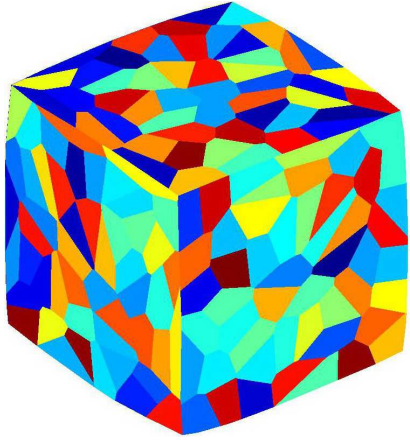}
\caption{Reconstruction of a unit cube from 300 random normal
  measurements with a uniform noise of radius 0.05. The reconstruction
  is obtained using the variational approach proposed in
  \cite{lachand}.}
\end{figure}


\begin{proposition}
\label{prop:polyhedra}
  Let $K$ be a convex polyhedron of $\Rsp^d$ with $k$ facets,
  non-empty interior and whose surface area $\Haus^{d-1}(\partial K)$
  equals one. Then one can construct a convex polyhedron $L_N$ such
  that
\[\Prob(\min_{x\in \Rsp^d} \dH(x+K,L_N) \leq \eta)\geq 1-p \]
from $N$ random normal measurements with $N \geq
\const(d,\rotund(\mu_K),k)\cdot \eta^{-2(d-1)} \log(1/p).$
\end{proposition}

\begin{proof}[Proof of Proposition~\ref{prop:polyhedra}]
  The surface area measure of $K$ can be written as $\mu_K =
  \sum_{1\leq i \leq k} a_i \delta_{\n_i}$ where the areas $(a_i)$ sum
  to one.  It is well-known that the empirical measure $\mu_{K,N}$
  constructed from a finitely support probability measures such as
  $\mu_{K}$ converges to the probability measure $\mu_K$ in the total
  variation distance with high probability. For instance, using
  Lemma~3 in \cite{devroye1983equivalence} we get
\begin{equation}
\Prob(\dTV(\mu_K,\mu_{K,N}) \geq \eps) \leq 3\exp(-N\eps^2/25),~\hbox{ assuming } \eps\geq  \sqrt{20k/N}.\label{eq:poly1}
\end{equation}
Now, let $\nu$ be an instance of $\mu_{K,N}$ such that
$\dTV(\mu_K,\nu) \leq \eps$. The measure $\nu$ can be written as $\nu
= \sum_{1\leq i\leq k} b_i \delta_{\n_i}$, and the assumption that the
total variation distance between $\mu_K$ and $\nu$ is at most $\eps$
can be rewritten as $\sum_{1\leq i\leq k} \abs{a_i - b_i} \leq
\eps$. The measure $\nu$ does not necessarily have zero mean, but one
can search for a perturbed measure $\bar{\nu} = \sum_{1\leq i\leq k}
\bar{b}_i \delta_{\n_i}$ with zero mean. More precisely, we let
\[ \bar{\nu} = \arg\min \{ \dTV(\nu,\bar{\nu}); \mean(\bar{\nu}) = 0 \}.\]
Solving this problem is equivalent to the minimization of a convex
functional on a finite-dimensional subspace. Moreover, since
$\dTV(\nu,\mu_K) \leq \eps$, we are sure that $\dTV(\nu,\bar{\nu})
\leq \eps$, so that $\dTV(\bar{\nu},\mu_K)\leq 2\eps$. Finally,
assuming $\eps$ small enough we have
\[ \dF(\bar{\nu},\mu_K) \leq \const(d) \dTV(\bar{\nu},\mu_K) \leq 2\eps \leq \frac{1}{2}\rotund(\mu_K).\]
Using Lemma~\ref{lemma:nondegstab}, this inequality ensures that
$\rotund(\bar{\nu}) \geq \rotund(\mu_K)/2 > 0$. By
Alexandrov's theorem, there exists a convex set $L_N$ whose
surface area measure $\mu_{L_N}$ coincides with $\bar{\nu}$, and whose
inradius and circumradius can be bounded in term of the weak rotundity
$\rotund(\mu_K)$. This allows us to apply Theorem~2.1 of \cite{hug} to
the sets $K$ and $L_N$ to show that
\begin{equation}
 \min_{x\in \Rsp^d} \dH(K+x,L_N) \leq \const(d,\rotund(\mu_K)) \eps^{\frac{1}{d-1}}
\label{eq:poly2}
\end{equation}
Combining Eqs \eqref{eq:poly1} and \eqref{eq:poly2}, we get
\begin{align*} \Prob\left(\min_{x\in \Rsp^d} \dH(K+x,L_N) \geq \eta\right) &\leq 
\Prob(\dTV(\mu_K,\mu_{K,N}) \geq c \cdot \eta^{d-1}) \\
&\leq \exp(-c \cdot N\eta^{2(d-1)})
\end{align*}
where $c$ depends on $d$ and $\rotund(\mu_K))$. This probability
becomes lower than $p$, and the assumption in Eq~\eqref{eq:poly1} is
satisfied, as soon as 
\begin{equation*}
N \geq \const(d,\rotund(\mu_K),k)\cdot
\eta^{-2(d-1)} \log(1/p).
\end{equation*}
\end{proof} 

\section{Discussion}

In this article, we introduced the convex-dual distance between
surface area measures. This distance is weaker than the
bounded-Lipschitz distance and is yet sufficient to control the
Hausdorff distance between convex bodies in term of the distance
between their surface area measures. This stability result has then
been used to deduce probabilistic reconstruction results
(Theorem~\ref{th:sampling}). The main open problem consists in
improving the exponent in the lower bound on the number of samples in
this theorem.

What would happen if we would have used the bounded-Lipschitz distance
in the probabilistic part of the proof of Theorem~\ref{th:sampling}
instead of the convex-dual distance? The lower bound on the number
$N$ of necessary normal measurements to get a Hausdorff error of
$\eps$ in the reconstruction would increase substantially:
\begin{equation}
N \geq \const\left(d,\rotund(\mu_K)\right)
\cdot \eta^{d(1-d) - 2d}\log(1/p).\label{eq:disc}
\end{equation}
In
particular, the exponents would become $N = \Omega(\eta^{-6})$ in
dimension two and $N = \Omega(\eta^{-12})$ in dimension three,
compared to $N=\Omega(\eta^{-5})$ and $N=\Omega(\eta^{-9})$ with our
analysis. This difference is due to the fact that the space $\BL_1$ of
functions on $\Sph^{d-1}$ that are $1$-Lipschitz and bounded by one is
\emph{much larger} than the space $\Conv_1$ of support function of
convex sets included in the unit ball. More precisely,
\[ \LebNum(\Conv_1,\eps) = \Theta\left(\eps^{\frac{1-d}{2}}\right)
\hbox{ while } \LebNum(\BL_1,\eps)=\Theta\left(\eps^{1-d}\right),
\]
where the constants in the $\Theta(.)$ notation only depend on the
ambient dimension.

It is therefore tempting to pursue in this direction, and to try to
consider a weaker dual distance between surface area measures,
i.e. defined with an even smaller space of functions. This idea is not
hopeless, as if one looks closely at the proofs of
Theorem~\ref{th:stab}, Lemma~\ref{lemma:nondegstab} and
Proposition~\ref{prop:perturb}, there are only a handful of probing
functions that are used to control the Hausdorff distance between two
convex bodies $K$ and $L$ as a function of their surface area measures
$\mu_K$ and $\mu_L$, and more precisely,
\[ \Conv_{K,L} = \left\{\h_K, \h_L\right\} \cup \left\{s_u; u\in \Sph^{d-1} \right\}
\cup \left\{\max(s_u,0); u \in \Sph^{d-1} \right\} \cup \{\h_{\B(0,1)}\},\]
where $s_u: x\mapsto \sca{u}{x}$. This set
of function is \emph{exponentially much smaller} than the set of
support functions of convex bodies included in the unit ball:
\[\LebNum(C_{K,L},\eps) \simeq \const(d) \cdot \eps^{1-d} \ll \LebNum(\Conv_1,\eps) \simeq \exp\left(\const(d)\cdot \eps^{\frac{1-d}{2}}\right).\]
However, turning this remark into an improvement of the probabilistic
analysis seems quite challenging, because in the probabilistic
setting, the second convex body $L_N$ is reconstructed from random
normal measurements and is itself random.

\subsection*{Acknowledgements.}
The authors would like to acknowledge the support of a grant from
Universit\'e de Grenoble (MSTIC GEONOR) and a grant from the French
ANR (Optiform, ANR-12-BS01-0007).  The authors would also like to
thank the members of the associated team ECR G\'eom\'etrie et Capteurs
CEA/UJF between LJK-UJF and CEA-LETI bringing up this problem to their
attention.


\bibliographystyle{acm}
\bibliography{minkowski}

\begin{thebibliography}{10}

\bibitem{Ale}
{\sc Alexandrov, A.}
\newblock On the theory of mixed volumes of convex bodies.
\newblock {\em Mat. Sb 3}, 45 (1938), 227--251.

\bibitem{bro}
{\sc Bronshtein, E.~M.}
\newblock $\epsilon$-entropy of convex sets and functions.
\newblock {\em Mat. Sb 17}, 3 (1976), 508--514.

\bibitem{carlier2004theorem}
{\sc Carlier, G.}
\newblock On a theorem of alexandrov.
\newblock {\em Journal of nonlinear and convex analysis 5}, 1 (2004), 49--58.

\bibitem{cheng1976regularity}
{\sc Cheng, S.-Y., and Yau, S.-T.}
\newblock On the regularity of the solution of the $n$-dimensional {M}inkowski
  problem.
\newblock {\em Communications on Pure and Applied Mathematics 29}, 5 (1976),
  495--516.

\bibitem{devroye1983equivalence}
{\sc Devroye, L.}
\newblock The equivalence of weak, strong and complete convergence in l1 for
  kernel density estimates.
\newblock {\em The Annals of Statistics\/} (1983), 896--904.

\bibitem{diskant}
{\sc Diskant, V.}
\newblock Bounds for the discrepancy between convex bodies in terms of the
  isoperimetric difference.
\newblock {\em Sib. Math. J. 13}, 4 (1972), 529--532.

\bibitem{diskant2}
{\sc Diskant, V.}
\newblock Bounds for the discrepancy between convex bodies in terms of the
  isoperimetric difference.
\newblock {\em Sib. Math. J. 13}, 4 (1972), 529--532.

\bibitem{dud}
{\sc Dudley, R.}
\newblock {\em Real analysis and probability}, vol.~74.
\newblock Cambridge Univ Pr, 2002.

\bibitem{gardner2006convergence}
{\sc Gardner, R.~J., Kiderlen, M., and Milanfar, P.}
\newblock Convergence of algorithms for reconstructing convex bodies and
  directional measures.
\newblock {\em The Annals of Statistics\/} (2006), 1331--1374.

\bibitem{gritz}
{\sc Gritzmann, P., and Hufnagel, A.}
\newblock On the algorithmic complexity of {M}inkowski's reconstruction
  theorem.
\newblock {\em J. Lond. Math. Soc. 59}, 3 (1999), 1081--1100.

\bibitem{guan1998monge}
{\sc Guan, P.}
\newblock Monge-ampere equations and related topics.
\newblock Course notes, 1998.

\bibitem{hug}
{\sc Hug, D., and Schneider, R.}
\newblock Stability results involving surface area measures of convex bodies.
\newblock {\em Rend. Circ. Mat. Palermo (2) Suppl.}, 70, part II (2002),
  21--51.

\bibitem{lachand}
{\sc Lachand-Robert, T., and Oudet, {\'E}.}
\newblock Minimizing within convex bodies using a convex hull method.
\newblock {\em SIAM J. Optim. 16}, 2 (2005), 368--379.

\bibitem{little}
{\sc Little, J.}
\newblock Extended gaussian images, mixed volumes, shape reconstruction.
\newblock In {\em Proc. Symposium on Computational Geometry\/} (1985), ACM,
  pp.~15--23.

\bibitem{min2}
{\sc Minkowski, H.}
\newblock Volumen und oberfl{\"a}che.
\newblock {\em Math. Ann. 57}, 4 (1903), 447--495.

\bibitem{schneider}
{\sc Schneider, R.}
\newblock {\em Convex bodies: the Brunn-Minkowski theory}.
\newblock Cambridge Univ Prss, 1993.

\bibitem{Luc}
{\sc Sprynski, N., Szafran, N., Lacolle, B., and Biard, L.}
\newblock Surface reconstruction via geodesic interpolation.
\newblock {\em Computer-Aided Design 40}, 4 (2008), 480--492.

\end{thebibliography}

\end{document}